\documentclass[11pt]{article}
\usepackage[margin=1in]{geometry}
\usepackage{amsmath}
\usepackage{amssymb}
\usepackage{amsthm}
\usepackage{algorithm}
\usepackage{algpseudocode}
\usepackage{booktabs}
\usepackage{multirow}
\usepackage{graphicx}
\usepackage{url}
\usepackage[hidelinks]{hyperref}
\usepackage{xurl}

\newtheorem{theorem}{Theorem}
\newtheorem{lemma}{Lemma}

\theoremstyle{definition}

\newcommand{\F}{\mathbb{F}}
\newcommand{\R}{\mathbf{R}}
\newcommand{\Rk}[2]{\mathbf{R}_{#1}(#2)}
\newcommand{\Cyct}[1]{\mathrm{C}_{#1}}
\newcommand{\Fullt}[1]{\mathrm{P}_{#1}}
\newcommand{\Trunct}[1]{\mathrm{T}_{#1}}
\newcommand{\Negat}[1]{\mathrm{C}^{-}_{#1}}
\newcommand{\GL}{\mathrm{GL}}
\newcommand{\PGL}{\mathrm{PGL}}
\newcommand{\Gal}{\mathrm{Gal}}
\newcommand{\Aut}{\mathrm{Aut}}
\newcommand{\Kappa}{\mathrm{K}}
\newcommand{\rank}{\operatorname{rank}}
\newcommand{\RREF}{\operatorname{RREF}}
\newcommand{\mmt}[3]{\langle #1,#2,#3\rangle}

\algnewcommand\algorithmicparallelfor{\textbf{parallel for}}
\algdef{SE}[PARFOR]{ParallelFor}{EndParallelFor}[1]%
{\algorithmicparallelfor\ #1\ \algorithmicdo}{\algorithmicend\ \algorithmicparallelfor}

\newif\ifanon
\anonfalse

\title{Automated Lower Bounds for Bilinear Complexity\\ over Finite Fields}
\ifanon
\author{}
\makeatletter
\renewcommand{\@maketitle}{%
  \newpage
  \null
  \vskip 2em%
  \begin{center}%
    {\LARGE \@title \par}%
  \end{center}%
  \par
\vskip 1.5em}
\makeatother
\else
\author{Chengu Wang\\ \texttt{wangchengu@gmail.com}}
\fi
\date{}

\begin{document}
\maketitle

\begin{abstract}
  We present a general, automated framework for proving lower bounds on the
  bilinear complexity (tensor rank) of multiplication problems over a finite field
  $\mathbb{F}_q$. The framework is parameterized only by the multiplication tensor and by a
  group of rank-preserving symmetries acting on one argument: it classifies
  the subspaces of the argument into orbits under the group, runs a dynamic program over
  the orbits combining four lower-bound techniques, and emits a proof certificate
  that a verifier rechecks, typically faster than the search.

  Instantiating the framework for matrix multiplication, we improve the lower bounds
  for three small formats over $\mathbb{F}_2$, most notably showing that
  the bilinear complexity of multiplying two $3 \times 3$ matrices over $\mathbb{F}_2$ is at
  least $20$, raising the bound of $19$ that had stood since Bl\"aser (2003).
  Instantiating it for polynomial multiplication, we obtain eighteen new lower
  bounds over $\mathbb{F}_2$ and $\mathbb{F}_3$, for the full product, cyclic
  convolution, and the truncated (modulo $x^N$) and negacyclic (modulo $x^N+1$)
  products.
  Every bound is backed by a machine-checkable certificate.
\end{abstract}

\section{Introduction}\label{sec:intro}

A \emph{bilinear} algorithm for a bilinear map $\beta\colon A\times B\to C$ over a
field $\F$ computes a fixed set of products of a linear form in the $A$-variables
with a linear form in the $B$-variables, and then reads off each coordinate of
the output as a linear combination of those products. The minimum number of such
products is the \emph{bilinear complexity} of $\beta$, and it equals the rank of
the order-$3$ tensor naturally associated with $\beta$. Two families of bilinear
maps have driven the subject: matrix multiplication and polynomial
multiplication. For matrix multiplication the central quantity is
$\R(\mmt{l}{m}{n})$, the rank of the tensor for multiplying an $l\times m$ by an
$m\times n$ matrix; for polynomial multiplication it is the rank of the tensor
for multiplying two polynomials.

Over small finite fields these ranks are notoriously hard to pin down, and even
for tiny formats the best known lower and upper bounds often disagree. This paper
develops one automated method that attacks both families, and uses it to push
several of these small cases.

\subsection{Lower bounds for matrix multiplication}

Strassen~\cite{strassen1969gaussian} showed $\R(\mmt{2}{2}{2})\le 7$ over $\mathbb{Z}$, and
Winograd~\cite{winograd1971multiplication} proved the matching lower bound over $\mathbb{Q}$.
Hopcroft and Kerr~\cite{hopcroft1971minimizing} then established several lower
bounds over $\F_2$, including $\Rk{\F_2}{\mmt{2}{2}{2}}\ge 7$ and
$\Rk{\F_2}{\mmt{2}{3}{3}}\ge 15$.
Ja'Ja' and Takche~\cite{ja1985improved} later proved
$\Rk{\F_2}{\mmt{3}{2}{n}}\ge\lceil\tfrac{105}{23}n\rceil$, which at $n=4$ gives
$\Rk{\F_2}{\mmt{2}{3}{4}}=\Rk{\F_2}{\mmt{3}{2}{4}}\ge 19$, along with
$\Rk{\F_2}{\mmt{n}{n}{n}}\ge 2n^2+\lceil\tfrac{3}{46}n\rceil-2$ and
$\Rk{\F_2}{\mmt{3}{3}{n}}\ge\lceil\tfrac{91}{16}n\rceil$, and
Bshouty~\cite{bshouty1989lower} proved $\Rk{\F_2}{\mmt{n}{n}{n}}\ge \tfrac{5}{2}n^2-o(n^2)$.
Bl\"aser established the influential bounds
$\R(\mmt{l}{m}{n})\ge lm+mn+l-m+n-3$ over algebraically closed
fields~\cite{blaser1999lower}, $\R(\mmt{n}{n}{n})\ge \tfrac{5}{2}n^2-3n$ over
arbitrary fields~\cite{blaser1999fivehalf}, and finally
$\R(\mmt{n}{m}{n})\ge 2mn+2n-m-2$ for $m\ge n\ge 3$~\cite{blaser2003complexity} over arbitrary fields,
which implies $\R(\mmt{3}{3}{3})\ge 19$.
Shpilka~\cite{shpilka2001lower} sharpened the asymptotics over small fields, proving
$\Rk{\F_2}{\mmt{n}{n}{n}}\ge 3n^2-o(n^2)$ and
$\Rk{\F_p}{\mmt{n}{n}{n}}\ge(\tfrac{5}{2}+\tfrac{3}{2(p^3-1)})n^2-o(n^2)$.
On the upper side
Laderman~\cite{laderman1976noncommutative} multiplied $3\times 3$ matrices with
$23$ products over $\mathbb{Z}$. Thus the rank of $\mmt{3}{3}{3}$ has remained trapped in $[19,23]$,
making it the prototypical open small format.

\subsection{Lower bounds for polynomial multiplication}

For the \emph{full} product of two degree-$(N-1)$ polynomials, the
output flattening immediately gives the lower bound $2N-1$; see also
Fiduccia and Zalcstein~\cite{fiduccia1977algebras}. This bound is attained
by evaluation--interpolation at $2N-1$ points of
$\mathbb{P}^1(\F_q)$ whenever $2N-1\le q+1$
\cite{toom1963complexity,kaminski1989multiplicative}.
Kaminski and Bshouty~\cite{kaminski1989multiplicative} proved that, when
$q/2<N-1\le q+1$, one has $ R_{\mathbb F_q}(P_N) = 3N-2-\lfloor q/2 \rfloor$.
For larger $N$, classical lower bounds come from the connection between
bilinear algorithms and linear codes developed by Brockett and
Dobkin~\cite{brockett1978optimal} and by Lempel and
Winograd~\cite{lempel1977new}. In particular, a rank-$r$ algorithm for
the full product yields a $q$-ary linear code of length $r$, dimension
$N$, and minimum distance at least $N$. Applying the Griesmer
bound~\cite{griesmer1960bound} (in its $q$-ary
form~\cite{solomon1965algebraically}) or the code tables of
Grassl~\cite{grassl2007codetables} therefore gives explicit lower
bounds. Some small instances have also been treated by exhaustive
search~\cite{barbulescu2012finding}.

A second strand is multiplication in a polynomial quotient ring
$\F_q[x]/(x^N-\gamma)$: \emph{cyclic} convolution ($\gamma=1$, modulo $x^N-1$),
the \emph{truncated} (or short) product ($\gamma=0$, modulo $x^N$), and the
\emph{negacyclic} product ($\gamma=-1$, modulo $x^N+1$), which coincides with the
cyclic one in characteristic $2$. The lower bounds for all three are governed by
the factorization of $x^N-\gamma$: Winograd~\cite{winograd1977some} proved that
multiplication modulo a polynomial with $k$ distinct irreducible factors needs at
least $2N-k$ products (a bound Alder and
Strassen~\cite{alder1981algorithmic} later extended to all associative algebras),
and Bl\"aser's characterization of algebras of minimal bilinear
complexity~\cite{blaser2005complete}, combined with the minimal-rank theory of
Averbuch, Galil, and
Winograd~\cite{averbuch1988classification,averbuch1991classification},
strengthens these on the local-algebra factors. The smallest cyclic and
truncated cases over $\F_2$ and $\F_3$ were settled exactly by the exhaustive
searches of Barbulescu et al.~\cite{barbulescu2012finding} and of
Covanov~\cite{covanov2019improved}.

\subsection{Our contributions}

We give a single framework, set up and instantiated per problem in
Section~\ref{sec:framework}, that proves bilinear-complexity lower bounds for any
multiplication tensor equipped with a group of first-argument symmetries. It
classifies the subspaces of the first argument up to symmetry
(Section~\ref{sec:orbits}), runs a
dimension-sweeping dynamic program with four lower-bound techniques
(Section~\ref{sec:techniques}), and produces certificates that an independent
verifier rechecks deterministically (Section~\ref{sec:certificates}).

Our main results concern matrix multiplication, where the framework improves the
best known lower bounds for three small formats over $\F_2$.

\begin{theorem}[Matrix multiplication]\label{thm:matrix}
  \[
    \Rk{\F_2}{\mmt{3}{3}{3}}\ge 20,\quad
    \Rk{\F_2}{\mmt{3}{3}{4}}\ge 25,\quad
    \Rk{\F_2}{\mmt{3}{4}{4}}\ge 29.
  \]
\end{theorem}

The headline result is $\Rk{\F_2}{\mmt{3}{3}{3}}\ge 20$, which raises the lower bound of $19$ that had
stood since Bl\"aser (2003)~\cite{blaser2003complexity}; the search finds this proof
in about $40$ minutes on a laptop and the certificate verifies in seconds. The same
certificate has since been re-verified independently by
Beuchert~\cite{beuchert_2026_20752782}, using a checker written from scratch in a
different programming language.

The tensor rank improvement over Bl\"aser's proof~\cite{blaser2003complexity} has two ingredients.
First, our framework fixes the ground field to $\F_2$, while his proof works over
arbitrary fields. Over $\F_2$ the candidate components of an optimal
decomposition form a small finite set, so the case analysis can branch on
every one. Second, the breadth of the induction: his proof passes through
some hand-crafted subspaces, whereas our argument runs
over more subspaces (all subspaces of the first $3\times 3$ matrix
argument), which makes the analysis tighter. The price is generality:
our result holds over $\F_2$ alone and for a single format,
whereas Bl\"aser's bound holds over arbitrary fields and for many formats.

For polynomial multiplication the same framework yields eighteen new lower bounds
over $\F_2$ and $\F_3$, for the full product and for the
three quotient families $\F_q[x]/(x^N-\gamma)$.

\begin{theorem}[Polynomial multiplication]\label{thm:poly}
  Write $\Rk{\F_q}{\Fullt{N}}$, $\Rk{\F_q}{\Cyct{N}}$, $\Rk{\F_q}{\Trunct{N}}$, and
  $\Rk{\F_q}{\Negat{N}}$ for the bilinear complexity over $\F_q$ of, respectively, the
  full product of two degree-$(N-1)$ polynomials, cyclic convolution (modulo $x^N-1$),
  the truncated product (modulo $x^N$), and the negacyclic product (modulo $x^N+1$),
  all taking two length-$N$ input coefficient vectors. Then
  \[
    \begin{aligned}
      \text{full:}       & \quad \Rk{\F_2}{\Fullt{6}}\ge 16,\ \ \Rk{\F_2}{\Fullt{7}}\ge 19,\ \ \Rk{\F_2}{\Fullt{8}}\ge 21,\ \ \Rk{\F_3}{\Fullt{6}}\ge 14;                                      \\
      \text{cyclic:}     & \quad \Rk{\F_2}{\Cyct{7}}\ge 13,\ \ \Rk{\F_2}{\Cyct{8}}\ge 19,\ \ \Rk{\F_2}{\Cyct{10}}\ge 22,\ \ \Rk{\F_3}{\Cyct{9}}\ge 19;                                         \\
      \text{truncated:}  & \quad \Rk{\F_2}{\Trunct{6}}\ge 13,\ \ \Rk{\F_2}{\Trunct{7}}\ge 16,\ \ \Rk{\F_2}{\Trunct{8}}\ge 19,\ \ \Rk{\F_2}{\Trunct{9}}\ge 21,\ \ \Rk{\F_2}{\Trunct{10}}\ge 24, \\
      & \quad \Rk{\F_2}{\Trunct{11}}\ge 27,\ \ \Rk{\F_3}{\Trunct{8}}\ge 17,\ \ \Rk{\F_3}{\Trunct{9}}\ge 19,\ \ \Rk{\F_3}{\Trunct{10}}\ge 21;                                \\
      \text{negacyclic:} & \quad \Rk{\F_3}{\Negat{9}}\ge 19.
    \end{aligned}
  \]
\end{theorem}

\begin{table}[!ht]
  \centering
  \small
  \setlength{\tabcolsep}{10pt}
  \begin{tabular}{lllll}
    \toprule
    family & problem                    & prev LB                                                                       & \textbf{our LB}        & prev UB                                                      \\
    \midrule
    \multirow{3}{*}{Matrix}
    & $\Rk{\F_2}{\mmt{3}{3}{3}}$ & $19$~\cite{blaser2003complexity}                                              & $\mathbf{20}$ & $23$~\cite{laderman1976noncommutative}                       \\
    & $\Rk{\F_2}{\mmt{3}{3}{4}}$ & $24$~\cite{blaser2003complexity}                                              & $\mathbf{25}$ & $29$~\cite{smirnov2013bilinear}                              \\
    & $\Rk{\F_2}{\mmt{3}{4}{4}}$ & $28$~\cite{blaser1999lower}                                                   & $\mathbf{29}$ & $38$~\cite{smirnov2013bilinear}                              \\
    \midrule
    \multirow{4}{*}{Full}
    & $\Rk{\F_2}{\Fullt{6}}$     & $15$~\cite{lempel1977new,brockett1978optimal,grassl2007codetables} & $\mathbf{16}$ & $17$~\cite{montgomery2005five}                               \\
    & $\Rk{\F_2}{\Fullt{7}}$     & $18$~\cite{lempel1977new,brockett1978optimal,grassl2007codetables} & $\mathbf{19}$ & $22$~\cite{montgomery2005five}                               \\
    & $\Rk{\F_2}{\Fullt{8}}$     & $20$~\cite{lempel1977new,brockett1978optimal,grassl2007codetables} & $\mathbf{21}$ & $26$~\cite{fan2007comments}                 \\
    & $\Rk{\F_3}{\Fullt{6}}$     & $12$~\cite{lempel1977new,brockett1978optimal,solomon1965algebraically}    & $\mathbf{14}$ & $15$~\cite{barbulescu2012finding}                                 \\
    \midrule
    \multirow{4}{*}{Cyclic}
    & $\Rk{\F_2}{\Cyct{7}}$      & $12$~\cite{degroote1983characterization,blaser2005complete}                   & $\mathbf{13}$ & $13$~\cite{wagh1983structured}                               \\
    & $\Rk{\F_2}{\Cyct{8}}$      & $16$~\cite{averbuch1991classification,blaser2005complete}                     & $\mathbf{19}$ & $22$~\cite{cenk2011multiplication} \\
    & $\Rk{\F_2}{\Cyct{10}}$     & $19$~\cite{averbuch1988classification,blaser2005complete}                     & $\mathbf{22}$ & $27$~\cite{randriambololona2012bilinear}                        \\
    & $\Rk{\F_3}{\Cyct{9}}$      & $18$~\cite{averbuch1991classification,blaser2005complete}                     & $\mathbf{19}$ & $27$~\cite{cenk2011multiplication} \\
    \midrule
    \multirow{9}{*}{Truncated}
    & $\Rk{\F_2}{\Trunct{6}}$    & $12$~\cite{averbuch1991classification,blaser2005complete}                     & $\mathbf{13}$ & $14$~\cite{cenk2011multiplication}                           \\
    & $\Rk{\F_2}{\Trunct{7}}$    & $14$~\cite{averbuch1991classification,blaser2005complete}                     & $\mathbf{16}$ & $18$~\cite{cenk2011multiplication}                           \\
    & $\Rk{\F_2}{\Trunct{8}}$    & $16$~\cite{averbuch1991classification,blaser2005complete}                     & $\mathbf{19}$ & $22$~\cite{cenk2011multiplication}                           \\
    & $\Rk{\F_2}{\Trunct{9}}$    & $18$~\cite{averbuch1991classification,blaser2005complete}                     & $\mathbf{21}$ & $27$~\cite{cenk2011multiplication}$^\ast$                    \\
    & $\Rk{\F_2}{\Trunct{10}}$   & $20$~\cite{averbuch1991classification,blaser2005complete}                     & $\mathbf{24}$ & $31$~\cite{cenk2011multiplication}                           \\
    & $\Rk{\F_2}{\Trunct{11}}$   & $22$~\cite{averbuch1991classification,blaser2005complete}                     & $\mathbf{27}$ & $36$~\cite{cenk2011multiplication}                           \\
    & $\Rk{\F_3}{\Trunct{8}}$    & $16$~\cite{averbuch1991classification,blaser2005complete}                     & $\mathbf{17}$ & $22$~\cite{cenk2011multiplication}                           \\
    & $\Rk{\F_3}{\Trunct{9}}$    & $18$~\cite{averbuch1991classification,blaser2005complete}                     & $\mathbf{19}$ & $27$~\cite{cenk2011multiplication}$^\ast$                          \\
    & $\Rk{\F_3}{\Trunct{10}}$   & $20$~\cite{averbuch1991classification,blaser2005complete}                     & $\mathbf{21}$ & $31$~\cite{cenk2011multiplication}                                 \\
    \midrule
    Negacyclic
    & $\Rk{\F_3}{\Negat{9}}$     & $18$~\cite{averbuch1991classification,blaser2005complete}                     & $\mathbf{19}$ & $27$~\cite{cenk2011multiplication}                           \\
    \bottomrule
  \end{tabular}

  \smallskip
  {\footnotesize\raggedright
    $^\ast$ We obtain a new bilinear complexity upper bound of $26$ for multiplication over $\mathbb{Z}[x]/x^9$ by flip-graph search~\cite{kauers2023flip} over $\F_2$ and lifting to $\mathbb{Z}$. See Table~\ref{tab:trunc} for the decomposition.
  \par}
  \caption{The 21 new lower bounds (3 matrix, 18 polynomial) proved in this paper.}
  \label{tab:summary}
\end{table}

In the $\Cyct{7}$ over $\F_2$ case, the new lower bound meets the best known
upper bound, settling $\Rk{\F_2}{\Cyct{7}}=13$. The truncated product at length $9$
also yields a new \emph{upper} bound: a rank-$26$ decomposition for multiplication
over $\mathbb{Z}[x]/x^9$, found by a flip-graph search adapted from Kauers and Moosbauer~\cite{kauers2023flip} over
$\F_2$ and lifted to $\mathbb{Z}$ (Table~\ref{tab:trunc}).

Table~\ref{tab:summary} collects all new lower bounds (Theorem~\ref{thm:matrix} and
Theorem~\ref{thm:poly}), placing each against the previous best lower bound and the
best known upper bound. Tables~\ref{tab:matrix}--\ref{tab:nega} in
Appendix~\ref{app:bound-tables} set these in the full context of each family,
including the many small cases where our framework reproduces the exact rank already
known.
The source code for the search and the verifier, together with all certificates, is
available
\ifanon
\else
at \url{https://github.com/wcgbg/tensor-rank-lower-bound}.
\fi
Appendix~\ref{app:results} (Table~\ref{tab:results}) reports computation details behind each new bound.

Several questions remain open. The exact rank of $\mmt{3}{3}{3}$ over $\F_2$ still
lies in $[20,23]$. Scaling to $\mmt{4}{4}{4}$ is blocked by the sheer number of
constraint orbits and will need stronger symmetry reduction. For polynomial
multiplication, the full, cyclic, truncated, and negacyclic bounds should extend to
larger $N$ with more computation. Finally, a natural next family is multiplication in
the extension field $\F_{q^n}$ over $\F_q$, computing $f\cdot g$ modulo an
irreducible polynomial of degree $n$ over $\F_q$.

The rest of the paper is organized as follows. Section~\ref{sec:toy} sketches the
whole framework on a single toy example. Section~\ref{sec:framework} sets up
subspaces, constraints, and their symmetries, instantiated for the matrix and
polynomial families. The two stages of the search follow: Section~\ref{sec:orbits}
enumerates the orbits, and Section~\ref{sec:techniques} gives the dynamic program
and its four lower-bound techniques. Section~\ref{sec:certificates} then details the
certificates and the verifier.

\section{A toy example: \texorpdfstring{$\Rk{\F_2}{\mmt{2}{2}{2}}\ge 7$}{rank over F2 of <2,2,2> >= 7}}\label{sec:toy}

Before the formal development, we sketch how the framework proves an easy, long-known
case of Hopcroft and Kerr~\cite{hopcroft1971minimizing}:
$2\times 2$ matrix multiplication over $\F_2$ needs at
least $7$ multiplications.

The method works with versions of the tensor $\mmt{2}{2}{2}$ in which
the first argument $X\in A=\F_2^{2\times2}$ ranges over a
\emph{subspace} of $A$ rather than all of $A$.
The framework runs in two stages.

\subsection{Stage 1: orbit enumeration}

Symmetries of the tensor that act on the
first input preserve rank, so subspaces related by a symmetry yield
tensors of the same rank, and it suffices to bound one representative
per orbit. For matrix multiplication the symmetries are $X\mapsto PXQ^{-1}$ with
$P,Q$ invertible, and $X\mapsto X^{\top}$. For $2\times 2$ matrices over $\F_2$
the enumeration (Section~\ref{sec:orbits}) finds ten orbits of
subspaces, ordered by dimension from the zero subspace (Orbit 0) up
to all of $A$ (Orbit 9), whose tensor is $\mmt{2}{2}{2}$ itself.

\subsection{Stage 2: dynamic program}

The second stage assigns each of the ten
orbits a rank lower bound, working upward in subspace dimension so that
each step may rely on the bounds already settled for smaller subspaces. Each
bound comes from one of four techniques (Section~\ref{sec:techniques}): flatten,
degenerate reduction, forced product, and substitution with backtracking. We now
walk through all ten orbits, grouped by technique.

\paragraph{Orbits 0, 1, 2, 3, 5 (flatten).} The base technique reads a bound off a
\emph{flattening}: regroup the tensor's three factors as two, view the result as a
matrix, and take its rank. Orbit 0, $X=\left(
  \begin{smallmatrix}0&0\\0&0
\end{smallmatrix}\right)$, is the zero tensor, so
$\R=0$. Orbit 1, $X=\left(
  \begin{smallmatrix}0&0\\0&a_{11}
\end{smallmatrix}\right)$, flattens to rank $2$,
so $\R\ge 2$. Orbits 2, 3, and 5, with $X=\left(
  \begin{smallmatrix}0&a_{01}\\a_{01}&0
\end{smallmatrix}\right)$, $\left(
  \begin{smallmatrix}0&0\\a_{10}&a_{11}
\end{smallmatrix}\right)$, and $\left(
  \begin{smallmatrix}0&a_{01}\\a_{10}&0
\end{smallmatrix}\right)$, each flatten to
rank $4$, so $\R\ge 4$.

\paragraph{Orbit 4, $X=\left(
    \begin{smallmatrix}0&a_{01}\\a_{01}&a_{11}
\end{smallmatrix}\right)$ (forced product).}
Flattening yields only $4$. The forced-product technique automates a substitution of
Hopcroft and Kerr~\cite{hopcroft1971minimizing}: two of the tensor's slices are single
products that any optimal decomposition may be assumed to compute literally, which uses
up $2$ of its terms; stripping those two terms leaves a residual whose every
$\F_2$-completion flattens to rank at least $4$. Hence $\R\ge 2+4=6$.

\paragraph{Orbit 6, $X=\left(
    \begin{smallmatrix}a_{00}&a_{01}\\a_{01}&a_{00}+a_{01}
\end{smallmatrix}\right)$ (substitution with backtracking).}
Flattening again gives only $4$, so we substitute. Fix an optimal decomposition; its
first components are nonzero linear forms in $a_{00},a_{01}$, of which there are just
three ($a_{00}$, $a_{01}$, $a_{00}+a_{01}$). With at least $4$ terms, by pigeonhole one
form repeats, and setting it to zero kills at least $2$ terms and lands in Orbit~2
(bound $4$). Hence $\R\ge 2+4=6$. The pigeonhole shortcut is for exposition only;
the framework automates such substitution arguments with a backtracking search
(Section~\ref{sec:backtracking}).

\paragraph{Orbits 7, 8 (degenerate reduction).} Passing to a subspace one dimension
smaller lands in an orbit already settled, whose bound is then inherited. Orbit 7,
$X=\left(
  \begin{smallmatrix}0&a_{01}\\a_{10}&a_{11}
\end{smallmatrix}\right)$, with $a_{10}$
set equal to $a_{01}$, and Orbit 8, $X=\left(
  \begin{smallmatrix}a_{00}&a_{01}\\a_{01}&a_{11}
\end{smallmatrix}\right)$, with
$a_{00}$ set to $0$, both land in Orbit~4, so $\R\ge 6$.

\paragraph{Orbit 9, $X=\left(
    \begin{smallmatrix}a_{00}&a_{01}\\a_{10}&a_{11}
\end{smallmatrix}\right)$ (substitution with backtracking).}
Here the subspace is all of $A$, and the tensor is $\mmt{2}{2}{2}$ itself. Fix an optimal decomposition
and take any of its first components, a nonzero linear form in the four $a_{ij}$.
Setting it to zero kills at least one term and lands in Orbit~7 or Orbit~8, each of
bound $6$. Since this holds for \emph{every} nonzero form, $\R\ge 1+6$, that is,
$ \Rk{\F_2}{\mmt{2}{2}{2}}\ge 7$.

Appendix~\ref{sec:toy-in-detail} works this example in full, with each orbit's
tensor and the backtracking search tree.

\section{Subspaces and symmetries}\label{sec:framework}

Fix a finite field $\F_q$ with $q=p^m$. A multiplication problem is a bilinear
map
\[
  \beta\colon A\times B\longrightarrow C,
\]
where $A=\F_q^{N_A}$ and $B=\F_q^{N_B}$ are its two input spaces and
$C=\F_q^{N_C}$ is its output space. After choosing bases $a_i$ of $A^*$,
$b_j$ of $B^*$, and $c_k$ of $C$, its structure tensor is
\[
  T = \sum_{i,j,k} t_{ijk}\, a_i \otimes b_j \otimes c_k
  \;\in\; A^* \otimes B^* \otimes C,
\]
where the structure constants $t_{ijk}\in\F_q$ encode $\beta$. The \emph{rank}
$\R(T)$ is the least $r$ such that
$T=\sum_{\lambda=1}^{r} u_\lambda\otimes v_\lambda\otimes w_\lambda$ with
$u_\lambda\in A^*$, $v_\lambda\in B^*$, and $w_\lambda\in C$; this is exactly the bilinear
complexity of $\beta$. We write $\Rk{\F_q}{T}$ when emphasizing the field $\F_q$,
and $\R(T)$ when it is clear from context.

We attack $\R(T)$ by restricting the first argument to a subspace $S\subseteq A$.
If $\operatorname{res}_S\colon A^*\to S^*$ sends a linear form to its restriction,
then the restricted tensor is
\[
  T_S=(\operatorname{res}_S\otimes\mathrm{id}_{B^*}\otimes\mathrm{id}_C)(T)
  \in S^*\otimes B^*\otimes C,
\]
and $\R(T)\ge\R(T_S)$.
Subspaces are described by \emph{constraints}: a constraint is a linear
functional on $A$, i.e.\ a vector in the dual $A^*\cong\F_q^{N_A}$, and a set of
constraints cuts $A$ down to the subspace on which all of them vanish.
Throughout, $d$ denotes the dimension of a subspace $S\subseteq A$ and
$d^{\perp}=N_A-d$ its codimension: $S$ needs $d^{\perp}$ independent constraints,
and all constraints vanishing on $S$ form a $d^{\perp}$-dimensional space, the
annihilator $S^{\perp}\subseteq A^*$. We represent $S$ by the reduced row echelon form of a
basis of $S^{\perp}$, which is a unique key for $S$; the
implementation manipulates subspaces exclusively through these constraint RREFs.

\subsection{Symmetries}

The leverage in our method comes from symmetries that act on $A$ and preserve
$\R(T)$. Abstractly, let $G$ be a group acting on $A$ by semilinear
bijections: each $g$ is additive with $g(\lambda u)=\phi_g(\lambda)\,g(u)$ for
some field automorphism $\phi_g\in\Gal(\F_q/\F_p)$, so both the $\F_q$-linear
sandwich and multiplication maps ($\phi_g=\mathrm{id}$) and the Galois symmetries
($\phi_g$ a Frobenius $a\mapsto a^p$) are allowed. We require that every
$g\in G$ extends to a rank-preserving symmetry of $T$; its action on the first
tensor factor $A^*$ is the contragredient of its action on $A$. Since a semilinear bijection sends $\F_q$-subspaces to
$\F_q$-subspaces, $G$ permutes the subspaces of $A$ and maps each $T_S$ to an
isomorphic tensor, so $\R(T_S)$ is constant on each $G$-orbit of subspaces. On
constraint vectors $G$ acts by the dual action, again semilinear, which is how
the implementation applies a group element to a stored subspace. Our
dynamic program needs only one representative per orbit.

The two facts the dynamic program rests on are the following.

\begin{lemma}[Monotonicity and orbit invariance]\label{lem:foundation}
  If $S'$ is a subspace of $S$, then $\R(T_{S})\ge\R(T_{S'})$. Moreover, for any rank-preserving
  first-argument symmetry $g\in G$ we have $\R(T_{S})=\R(T_{g(S)})$, so $\R(T_S)$
  depends only on the $G$-orbit of the subspace.
\end{lemma}

\begin{proof}
  A bilinear algorithm with $r$ products that computes the map of $T_S$ also computes
  its restriction to any smaller first-argument subspace with the same $r$ products, so
  imposing further constraints cannot raise the rank; this is monotonicity. For
  invariance, a rank-preserving symmetry extending the first-input action $g$ carries
  $T_S$ to $T_{g(S)}$ by an invertible change of variables, which leaves the rank
  unchanged.
\end{proof}

\subsection{The query/store factorization}\label{sec:query-store}

The implementation never materializes $G$ as a list. Instead it factors the group $G$
into two subsets of $G$, a \emph{query} set $\Kappa$ and a \emph{store} set
$\Sigma$, such that the product set $\Kappa^{-1}\cdot\Sigma$ covers $G$ so that every
orbit member is reached. Neither $\Kappa$ nor $\Sigma$ needs to be a subgroup or be
closed under inversion. This is the meet-in-the-middle, or square-root, trick. To
test whether a query subspace $S$ lies in the orbit of a stored representative
$c$, we precompute and hash the RREF of $\sigma(c)$ for every
$\sigma\in\Sigma$ (the \emph{store images}), and at query time probe the hash with the
RREF of $\kappa(S)$ for every $\kappa\in\Kappa$. A hit $\kappa(S)\equiv \sigma(c)$
certifies $S$ and $c$ as $G$-equivalent and yields the witness pair
$(\kappa,\sigma)$, from which the canonical form is recovered as
\begin{equation}\label{eq:witness}
  c \;\equiv\; \sigma^{-1}\bigl(\kappa(S)\bigr).
\end{equation}
Storage grows by a factor $|\Sigma|$ and a query costs $|\Kappa|$
probes, with $|\Kappa|\cdot|\Sigma|\gtrsim |G|$; choosing
$|\Kappa|\approx|\Sigma|\approx\sqrt{|G|}$ trades square-root space for
square-root time.

The framework is fixed once we name, for each problem, the tensor $T$ and the
factorization $(\Kappa,\Sigma)$ of its first-argument symmetry group; the next
three subsections do so.

\subsection{Matrix multiplication}\label{sec:matrix}

Let
\[
  A=\F_q^{l\times m},\qquad B=\F_q^{m\times n},\qquad
  C=\F_q^{l\times n}.
\]
For matrices $X\in A$ and $Y\in B$, write $Z=XY\in C$. Let $x_{ij}\in A^*$
and $y_{jk}\in B^*$ be the coordinate functionals, and let $e_{ik}\in C$ be
the matrix units. The tensor for matrix multiplication is
\[
  \mmt{l}{m}{n} = \sum_{i=0}^{l-1}\sum_{j=0}^{m-1}\sum_{k=0}^{n-1}
  x_{ij}\otimes y_{jk}\otimes e_{ik}
  \in A^*\otimes B^*\otimes C.
\]
For invertible $P\in\GL_l$, $Q\in\GL_m$, and $R\in\GL_n$, the changes
$X\mapsto PXQ^{-1}$, $Y\mapsto QYR^{-1}$, and $Z\mapsto PZR^{-1}$ give a
rank-preserving isomorphism of the multiplication map (the \emph{sandwich}
symmetry). The rank is also invariant under
cyclically permuting the three factors, $\R(\mmt{l}{m}{n})=\R(\mmt{m}{n}{l})$ (the
\emph{cyclic} symmetry). For a square format $l=m=n$, factor permutation and
matrix transposition also give an order-two rank-preserving symmetry whose
action on the first input is $X\mapsto X^{\!\top}$ (the \emph{transpose}
symmetry). Projected onto the first input, the symmetry group acting on
$A=\F_q^{l\times m}$ is
\[
  G \;=\; (\GL_l(\F_q)\times \GL_m(\F_q)) \rtimes C_2,
\]
with $(P,Q)$ acting by $X\mapsto PXQ^{-1}$ and the $C_2$ by transposition $X\mapsto X^{\!\top}$ (present
only when $l=m=n$). The meet-in-the-middle splits $G$ by side:
the store set is the right multiplications $\Sigma=\{X\mapsto XQ^{-1}:Q\in\GL_m(\F_q)\}$ and
the query set is the left multiplications $\Kappa=\{X\mapsto LX:L\in\GL_l(\F_q)\}$,
extended by the transpose when $l=m=n$. Scalar matrices act trivially on subspaces,
so each side is taken modulo the center, and both then have order about $\sqrt{|G|}$.
Equivalence of dimension-$d$ subspaces under the subgroup without transpose is
exactly the tensor-isomorphism problem on $l\times m\times d$ tensors (stack the
$d$ basis matrices as slices), which sits above graph
and code equivalence in the isomorphism hierarchy~\cite{grochow2023complexity}.

\subsection{Full polynomial multiplication}

Let $A=B=\F_q[x]_{<N}$ be the two input spaces and
$C=\F_q[x]_{<2N-1}$ the output space. With $a_i\in A^*$ and $b_j\in B^*$ the
coefficient functionals and $c_k\in C$ the monomial basis, multiplying two
polynomials of degree $<N$ gives
\[
  \Fullt{N} = \sum_{i,j=0}^{N-1} a_i\otimes b_j\otimes c_{i+j},
  \qquad \Fullt{N}\in A^*\otimes B^*\otimes C,
\]
the ordinary convolution. Its first-argument symmetry group is the projective
semilinear group
\[
  G_{\mathrm{full}} = \mathrm{P\Gamma L}_2(\F_q)
  = \PGL_2(\F_q)\rtimes \Gal(\F_q/\F_p),
\]
acting on degree-$(N-1)$ binary forms and generated by the substitutions
$x\mapsto x+c$, the reversal $x\mapsto 1/x$ (which
swaps $a_i\leftrightarrow a_{N-1-i}$), the scalings $x\mapsto gx$, and the
Frobenius $a\mapsto a^p$. We factor it as
$\Sigma=\PGL_2(\F_q)$ (the store side, of size $q(q^2-1)$) and
$\Kappa=\Gal(\F_q/\F_p)$ (the query side, of size $m$).

\subsection{Polynomial multiplication modulo \texorpdfstring{$x^N-\gamma$}{x to the N minus gamma}}

Fix a scalar $\gamma\in\F_p\subseteq\F_q$ and let
$R=\F_q[x]/(x^N-\gamma)$. This assumption includes the three cases
$\gamma\in\{0,1,-1\}$ considered below and ensures that every element of
$\Gal(\F_q/\F_p)$ fixes $\gamma$. Take the two input spaces and the output space
to be three copies $A,B,C$ of $R$, let $a_i\in A^*$
and $b_j\in B^*$ be the coefficient functionals, and let $c_k\in C$ be the
monomial basis. Multiplication in $R$ reduces the convolution by the rule
$x^N\equiv\gamma$,
\[
  \sum_{i,j=0}^{N-1} a_i\otimes b_j\otimes w_{ij},\qquad
  w_{ij}=
  \begin{cases} c_{i+j},           & i+j<N,    \\[2pt]
    \gamma\,c_{i+j-N}, & i+j\ge N,
  \end{cases}
  \qquad T\in A^*\otimes B^*\otimes C.
\]
Three families specialize this tensor: \emph{cyclic} convolution $\Cyct{N}$ for
$\gamma=1$, the \emph{truncated} (or short) product $\Trunct{N}$ for $\gamma=0$
(the local algebra $\F_q[x]/(x^N)$, in which the overflow terms vanish and $x$ is
nilpotent), and the \emph{negacyclic} product $\Negat{N}$ for $\gamma=-1$. Over
characteristic $2$, $-1=1$, so $\Negat{N}=\Cyct{N}$.

The same three kinds of symmetry act on $A=R$ for every such $\gamma$: multiplication by a
unit $u\in R^*$ (the shifts $x^r$ are the special case $u=x^r$, valid for
$\gamma\ne 0$), the $\F_q$-algebra automorphisms $\Aut(R)$, and the Galois group of
$\F_q/\F_p$. The base-field scalars $\F_q^*$ act trivially on subspaces, so the group
acting on subspaces is
\[
  G = \bigl(R^*\rtimes \Aut(R) \rtimes \Gal(\F_q/\F_p)\bigr)/\F_q^*,
\]
with store side $\Sigma=R^*/\F_q^*$ and query side
$\Kappa=\Aut(R)\rtimes\Gal(\F_q/\F_p)$, the identical factorization in all three
cases. Only the ring structure, and hence the sizes of the two sides, depends on
$\gamma$.
For $\gamma=\pm1$ the structure of $R$ depends on whether $p\mid N$. If $p\nmid N$,
then $x^N-\gamma$ is squarefree and the Chinese remainder theorem splits $R$ into
field factors (for $\gamma=1$ indexed by the $q$-cyclotomic cosets of
$\mathbb{Z}/N$, for $\gamma=-1$ by the cosets of the odd residues modulo $2N$). If
$p\mid N$, write $N=N'p^s$ with $p\nmid N'$; then $x^N-\gamma=(x^{N'}-\gamma)^{p^s}$
and $R$ splits into local rings rather than fields. Either way $R^*$ is far larger
than the $N$ shifts, which is precisely the extra symmetry the framework exploits.
The implementation never uses this decomposition: it enumerates $R^*$ as the
elements whose multiplication matrix is invertible and $\Aut(R)$ as the elements $y$ with
$y^N=\gamma$ and $\F_q[y]=R$, both valid for every such $\gamma$ and $N$.
For $\gamma=0$ the algebra is local: a polynomial is a unit iff
its constant term is nonzero, so $|\Sigma|=q^{N-1}$, and an automorphism sends $x$ to
any nilpotent generator $y_1x+\dots+y_{N-1}x^{N-1}$ with $y_1\ne 0$, giving
$|\Aut(R)|=(q-1)q^{N-2}$.

\medskip

Table~\ref{tab:symmetry} contrasts the symmetry structures; the three quotient
families share one row, differing only in the constant $\gamma$.

\begin{table}[h]
  \centering
  \setlength{\tabcolsep}{6pt}
  \begin{tabular}{lll}
    \toprule
    problem                            & store $\Sigma$        & query $\Kappa$                                \\
    \midrule
    matrix-mult $\mmt{l}{m}{n}$        & $\GL_m$ (right mult.) & $\GL_l$ (left mult.), and transpose ($l=m=n$) \\
    full-poly-mult                     & $\PGL_2(\F_q)$        & $\Gal(\F_q/\F_p)$                             \\
    poly-mult $\F_q[x]/(x^N{-}\gamma)$ & $R^*/\F_q^*$          & $\Aut(R)\rtimes\Gal(\F_q/\F_p)$               \\
    \bottomrule
  \end{tabular}
  \caption{First-argument symmetry factorizations; the last row covers the cyclic
    ($\gamma=1$), truncated ($\gamma=0$), and negacyclic ($\gamma=-1$) products. In every family the
    query and store elements are numbered consecutively, so a certificate can name
  a symmetry by a single integer.}
  \label{tab:symmetry}
\end{table}

\section{Orbit enumeration}\label{sec:orbits}

The first stage enumerates exactly one canonical representative per orbit of
subspaces of $A$. Subspaces are cut down one dimension at a time
and deduplicated with the same meet-in-the-middle
trick (Section~\ref{sec:query-store}).

Representatives are generated codimension by codimension, in lexicographic
order of their RREF in the dual $A^*$. A codimension-$d^{\perp}$ subspace is formed by
appending to a codimension-$(d^{\perp}-1)$ representative one new constraint whose leading
pivot lies strictly above the existing pivots; candidates with a nonzero
entry at an existing pivot are skipped because Gauss--Jordan elimination would
reduce them against the current basis. Because candidates are visited in global
lexicographic order, the first time an orbit is seen its representative is
automatically the lexicographically least, so no separate minimality test is
needed.

Deduplication uses a hash set of RREF keys.
For each kept representative $c$ we insert the RREF of every store image
$\sigma(c)$, $\sigma\in\Sigma$; for each new candidate $S$ we probe with
the RREF of every query image $\kappa(S)$, $\kappa\in\Kappa$, and
discard $S$ if any probe hits, since a hit means $S\equiv (\kappa^{-1}\sigma)(c)$
lies in $c$'s orbit.
Algorithm~\ref{alg:enum}, in Appendix~\ref{app:enum}, states the procedure.

\section{Dynamic program with four lower-bound techniques}\label{sec:techniques}

The second stage is a dynamic program that sweeps subspace
dimension from $0$ (zero subspace) to $N_A$ ($A$ itself); this is the reverse of
the enumeration order of Section~\ref{sec:orbits}. At
each dimension it processes all orbits of that
dimension in parallel, then inserts their freshly computed bounds into the orbit
map so that the next (larger) dimension can consult them. This ordering is sound
because the two techniques that look anything up (degenerate reduction and
substitution with backtracking) only ever query \emph{strictly smaller}
subspaces, which have already been settled. Each orbit tries four techniques in
this order: flatten, degenerate reduction, forced product, and substitution with
backtracking; later techniques run only if the earlier ones have not already met
the goal. The order is roughly one of increasing cost, and the last technique
(Section~\ref{sec:backtracking}) is both the most powerful and the main novelty
of the framework.

\subsection{Flatten}\label{sec:flatten}

The flatten technique computes the three flattening matrix ranks of the
constrained tensor and takes the maximum; it is the base case and needs no other
orbit.

\subsection{Degenerate reduction}\label{sec:degenerate}

Degenerate reduction imposes one more independent constraint on the current
subspace, producing a strictly smaller subspace whose orbit was processed
at the previous dimension; by Lemma~\ref{lem:foundation} its bound is a
valid bound for the current orbit. We scan candidate constraints, look up each
such smaller subspace in the orbit map, and keep the best.

\subsection{Forced product}\label{sec:forced-product}

The forced-product technique automates the Hopcroft--Kerr
substitution~\cite{hopcroft1971minimizing}, which rests on the following lemma. The
lemma is their Lemma~2; we reproduce its short proof for completeness.

\begin{lemma}[Hopcroft and Kerr~\cite{hopcroft1971minimizing}]\label{lem:forced-product}
  Let $f_0,\dots,f_{t-1}$ be bilinear forms of which $f_0,\dots,f_{s-1}$ are linearly
  independent and each is a single product (a form of rank $1$). If the whole set can
  be computed with $r$ multiplications, then it can be computed with $r$
  multiplications, $s$ of which are exactly $f_0,\dots,f_{s-1}$.
\end{lemma}

\begin{proof}
  Take any bilinear algorithm with products $p_0,\dots,p_{r-1}$ computing the $f_i$,
  each a fixed $\F_q$-linear combination of the products. We exchange products for the
  $f_i$ one at a time, maintaining the invariant that after step $i$ the multiset of
  products is $f_0,\dots,f_{i-1},$ together with $r-i$ of the original $p_j$, and that
  its span still contains every one of $f_0,\dots,f_{t-1}$.
  At step $i$, write $f_i$ in terms of the current
  products. Since $f_i$ is linearly independent of $f_0,\dots,f_{i-1}$, some
  \emph{original} product $p_j$ (not one of the already-exchanged $f_0,\dots,f_{i-1}$)
  must occur in this expression with a nonzero coefficient $\alpha$. Then
  $p_j=\alpha^{-1}(f_i-\text{rest})$ lies in the span of $f_i$ and the other current
  products, so replacing $p_j$ by $f_i$ does not shrink the span; and $f_i$, being of
  rank $1$, is itself a legal product. After $s$ steps the algorithm has $r$ products,
  $s$ of which are exactly $f_0,\dots,f_{s-1}$.
\end{proof}

Group the tensor along the $C$ factor into $t$ slices; the slices whose
$A^*\otimes B^*$ part has matrix rank $1$ are single products. We greedily take a
maximal linearly independent set of $s$ of them and, by
Lemma~\ref{lem:forced-product}, assume an optimal decomposition computes them
literally. Stripping these $s$ terms leaves $s(t-s)$ unknown field coefficients in
the residual; we enumerate all $q^{s(t-s)}$ assignments, flatten each residual, and
take the minimum, so the bound is $s$ plus that minimum. We repeat for all three
cyclic orientations of the tensor and keep the best. For soundness the minimum must
range over \emph{every} assignment, each coefficient running over all $q$ values; we
skip the technique when $q^{s(t-s)}$ is too large to enumerate. The proof records only which of the three
orientations won. Note that Lemma~\ref{lem:forced-product} holds over any field; it
is the exhaustive enumeration of the $q^{s(t-s)}$ residuals that makes the technique
\emph{effective} only over a small finite field.

\subsection{Substitution with backtracking}\label{sec:backtracking}

The substitution method is a powerful inductive technique in bilinear complexity
lower bound proofs, introduced by Pan~\cite{pan1966methods}.
We state the iterated form the search uses and include
its short proof for completeness.

\begin{lemma}[Substitution]\label{lem:substitution}
  Fix an optimal decomposition $T_S=\sum_{\lambda}u_\lambda\otimes v_\lambda\otimes
  w_\lambda$ with $\R(T_S)$ terms, where $u_\lambda\in S^*$ is nonzero for every
  $\lambda$. Suppose that repeatedly setting a linear form to zero
  (each step imposing that form as a new constraint and discarding the terms whose
  first component vanishes on the smaller subspace) removes $\delta$ terms in total and
  reaches the more-constrained tensor $T_{S'}$. Then $\R(T_S)\ge\delta+\R(T_{S'})$.
\end{lemma}

\begin{proof}
  Setting a linear form to zero constrains the first argument to the hyperplane of the
  current subspace on which the form vanishes. The terms whose first component vanishes
  on that hyperplane restrict to zero and drop out; the
  surviving terms remain a valid decomposition of the constrained tensor. So a single
  step that discards $e$ terms gives $\R(T_S)\ge e+\R(T_{S''})$ for the one-smaller
  subspace $S''$. Iterating along the chain and summing the discarded terms yields the
  claim.
\end{proof}

So if a chain of substitutions removing $\delta$ terms lands in an orbit of
already-certified bound $b$ with $\delta+b\ge$ target, then $\R(T_S)\ge$ target. In
the hand proofs the art lies in arguing that a good chain must exist, typically by
a pigeonhole argument on the first components of an optimal decomposition.
Substitution with backtracking replaces this ingenuity with an exhaustive case analysis
over those components, organized as a depth-first search in which \emph{every}
branch must reach the target.

\paragraph{Canonical forms.} Let $d$ be the dimension of the current
subspace $S\subseteq A$. Substituting a form is unaffected by scaling it or by
adding to it a constraint already in $S^{\perp}$, so call two forms \emph{equivalent} when
they are proportional modulo $S^{\perp}$. The technique enumerates one canonical form
per class: the vectors of $A^*$ that have a zero in every pivot column of the
stored $\RREF(S^\perp)$, normalized to leading coefficient $1$. There are $(q^d-1)/(q-1)$ of
them; fix an order $f_0,f_1,\dots$. By the hypothesis of
Lemma~\ref{lem:substitution}, every term of an optimal decomposition of $T_S$ has
a first component equivalent to exactly one $f_i$, and imposing the constraint
$f_i=0$ kills every term whose component is equivalent to $f_i$.

\paragraph{The search tree.} Suppose the earlier techniques have certified
$\R(T_S)\ge k$ and the target is $k+1$. Toward a contradiction, assume
$\R(T_S)=k$. An optimal decomposition then has exactly $k$
terms; list the canonical forms of its first components in non-decreasing index
order, so that every possible multiset of components yields exactly one list.
The DFS enumerates the prefixes of this unknown list: a node at
depth $j$ is a chain $f_{i_1},\dots,f_{i_j}$ with $i_1\le\dots\le i_j$,
representing the case that $j$ distinct terms of the decomposition have these
components. At each node the technique tries every subset of the chain that
contains the newest form $f_{i_j}$ (subsets without it were tried at an
ancestor). Imposing the subset's distinct forms as extra constraints kills at
least one term per chain entry in the subset (entries repeating a form are
distinct terms killed by the same substitution) and lands in the smaller subspace
of $S$ on which the subset's forms also vanish, and
hence already processed by the sweep; one orbit-map query retrieves its certified bound $b$ together with
the witness pair of Equation~\eqref{eq:witness}, and
Lemma~\ref{lem:substitution} gives $\R(T_S)\ge s+b$, where $s$ is the number of
chain entries in the subset. If some subset reaches $s+b\ge k+1$, the node is a
\emph{successful leaf} and the branch closes; otherwise the DFS deepens, and
since the next entry of the list can be any $f_i$ with $i\ge i_j$, \emph{all}
such children must close. Chains never grow past length $k$, the length of the
whole list, which caps the DFS depth. If every branch closes, then the true
decomposition's own prefix chain passes through some successful leaf; that leaf's
subset consists of components of genuinely distinct terms, so its inequality
yields $\R(T_S)\ge k+1$, contradicting $\R(T_S)=k$. Hence $\R(T_S)\ge k+1$. The
technique then re-runs with the target raised by one until it fails and keeps the
last success; aiming each run only one unit higher lets branches close at shallow
depths. Note how a form repeated in a chain makes a single substitution remove
several terms: the enumeration rediscovers, mechanically, the pigeonhole
arguments of the hand proofs. Orbit~6 of Appendix~\ref{sec:toy-in-detail}
displays a complete search tree, every branch of which closes at a chain with a
repeated form.

\paragraph{Scaling.} Three optimizations make the search practical, and none touches
the exhaustiveness of the case analysis. The top level of the DFS (the choice of
$f_{i_1}$) runs in parallel, and the whole technique is abandoned for the orbit
as soon as any branch fails. A thread-local cache memoizes the orbit-map answer
for each smaller subspace, which recurs across many chains; when the cache
outgrows its memory budget it is halved stochastically. And a global step limit
aborts oversized searches; an aborted branch counts as failed, so the limit
trades away completeness of the \emph{search}, never soundness of a reported
bound. Whenever the technique fails, the dynamic program falls back to the best
bound from the other techniques. On success, every successful leaf contributes
one record to the proof trace; Section~\ref{sec:certificates} describes the
trace format and its replay by the verifier.

\section{Proof certificates and verification}\label{sec:certificates}

Every bound is emitted as a certificate that a verifier rechecks.
The verifier is simpler than the search program, and it runs faster and uses less memory.

For each orbit, the certificate records
its canonical subspace (as the constraint RREF), the certified lower bound, and a
proof by one of the four techniques.
The bound for the unconstrained orbit is ``the'' bound for the problem.
Each proof stores only what the
verifier cannot cheaply recompute:
\begin{itemize}
  \item \textbf{Flatten:} nothing; the verifier recomputes the flattening rank.
  \item \textbf{Forced product:} the winning cyclic orientation, an integer in
    $\{0,1,2\}$.
  \item \textbf{Degenerate reduction:} the extra constraint, together with the witness pair
    $(\kappa,\sigma)$ of Equation~\eqref{eq:witness} that identifies the orbit of
    the resulting smaller subspace.
  \item \textbf{Substitution with backtracking:} one record per successful leaf of the DFS of
    Section~\ref{sec:backtracking}, in depth-first order: the leaf's depth, the
    subset of the chain that met the target (a bit mask), and the witness pair
    of the orbit it reduced to. Because the DFS is deterministic, the
    verifier replays the identical tree without searching. It regenerates the
    canonical forms, walks the same branches, treats a node whose depth matches
    the next record's as a recorded leaf, rechecking that the subset size plus
    the looked-up bound of the witnessed orbit meets the claimed bound, and
    descends into every child elsewhere, so a trace that skips any case fails
    to verify.
\end{itemize}

The verifier replays the dimension sweep of Section~\ref{sec:techniques},
checking the proof for each orbit in order from $0$-dimensional to
$N_A$-dimensional subspaces.
For flattening and forced product, the verifier recomputes the claimed lower bound directly,
without relying on any other orbit. For degenerate reduction and substitution with backtracking,
it reconstructs each claimed reduction and checks it against
the certified bound of the resulting strictly smaller subspace,
which has already been verified earlier in the sweep.
Thus, by induction on the subspace dimension,
every accepted orbit bound is valid. Finally,
if the verifier accepts the certificate for the unconstrained orbit (dimension $N_A$)
with bound $r$, then $\R(T)\ge r$.

\section*{Acknowledgements}
\ifanon
\else
We thank Markus Bl\"aser for pointing out the papers \cite{ja1985improved} and
\cite{shpilka2001lower}. We also thank Youming
Qiao for discussions on the tensor isomorphism problem and for reviewing a draft of
this paper.
\fi

\section*{AI Disclosure}
Generative AI tools materially contributed to this work. We used Claude, ChatGPT,
and Gemini for background research and literature search, and for exploring ideas,
such as generalizing the technique from matrix
multiplication to polynomial multiplication. We wrote the matrix multiplication lower bound
search code by hand and used Claude Code to generalize it to the polynomial
multiplication families; we also used Claude and ChatGPT to check mathematical derivations
and to draft and edit this paper. The author verified the correctness and originality
of all content, including references.

\bibliography{refs}

\appendix

\section{Complete bound tables by family}\label{app:bound-tables}

Tables~\ref{tab:matrix}--\ref{tab:nega} give the full lower- and upper-bound picture
for each problem family. The boxed entries are the new lower bounds proved in this
paper and collected in Table~\ref{tab:summary}.
In characteristic $2$ the negacyclic product coincides with cyclic
convolution ($x^N+1=x^N-1$), so Table~\ref{tab:nega} lists only $\F_3$.

In six entries the previously known lower bound exceeds ours by one: cyclic
convolution at $N=5$ over both $\F_2$ and $\F_3$ and at $N=9$ over $\F_2$
(Table~\ref{tab:cyc}), the truncated product at $N=5$ over $\F_2$ and at $N=4$
over $\F_3$ (Table~\ref{tab:trunc}), and the negacyclic product at $N=5$ over
$\F_3$ (Table~\ref{tab:nega}). These entries share a pattern: the previous bound
is tight (it meets the best known upper bound, so the rank is known exactly), and
it comes either from the algebra-specific minimal-rank theory of
Bl\"aser~\cite{blaser2005complete} and of Averbuch, Galil, and
Winograd~\cite{averbuch1991classification}, which exploits the multiplicative
structure of the quotient algebra (for cyclic $N=9$ over $\F_2$ combined with
the code-embedding bound on the field factor $\F_{2^6}$~\cite{lempel1977new,grassl2007codetables}),
or from the exhaustive
computer searches of
Barbulescu, Detrey, Estibals, and Zimmermann~\cite{barbulescu2012finding} and of
Covanov~\cite{covanov2019improved}, which enumerate all optimal decompositions of
one fixed small tensor. Our
substitution-based techniques do not subsume those arguments, and in each of these
cases they stop one short of the known rank; the ``our LB'' column always reports
what our framework proves, not the stronger inherited bound.

The tables also serve as regression tests for the implementation: in no entry does a
bound we prove exceed the best known upper bound, and in many small cases it matches
the exact rank already known.

\begin{table}[!ht]
  \centering
  \setlength{\tabcolsep}{14pt}
  \begin{tabular}{llll}
    \toprule
    format          & prev LB                                                          & our LB       & prev UB                                                      \\
    \midrule
    $\mmt{2}{2}{2}$ & $7$~\cite{hopcroft1971minimizing}                                & $7$          & $7$~\cite{strassen1969gaussian}                              \\
    $\mmt{2}{2}{3}$ & $11$~\cite{hopcroft1971minimizing}                               & $11$         & $11$ ($\langle 2, 2, 2:7\rangle + \langle 2, 2, 1:4\rangle$) \\
    $\mmt{2}{2}{4}$ & $14$~\cite{hopcroft1971minimizing}                               & $14$         & $14$ ($\langle 2, 2, 2:7\rangle + \langle 2, 2, 2:7\rangle$) \\
    $\mmt{2}{3}{3}$ & $15$~\cite{hopcroft1971minimizing}                               & $15$         & $15$~\cite{hopcroft1971minimizing}                           \\
    $\mmt{2}{3}{4}$ & $19$~\cite{ja1985improved}                                       & $19$         & $20$~\cite{hopcroft1971minimizing}                           \\
    $\mmt{3}{3}{3}$ & $19$~\cite{blaser2003complexity}                                 & $\boxed{20}$ & $23$~\cite{laderman1976noncommutative}                       \\
    $\mmt{3}{3}{4}$ & $24$~\cite{blaser2003complexity}                                 & $\boxed{25}$ & $29$~\cite{smirnov2013bilinear}                              \\
    $\mmt{3}{4}{4}$ & $28$~\cite{blaser1999lower}                                      & $\boxed{29}$ & $38$~\cite{smirnov2013bilinear}                              \\
    \bottomrule
  \end{tabular}
  \caption{Matrix multiplication tensor rank over $\F_2$; boxed entries are new.}
  \label{tab:matrix}
\end{table}

\begin{table}[!ht]
  \centering
  \setlength{\tabcolsep}{6pt}
  \begin{tabular}{lllllll}
    \toprule
    & \multicolumn{3}{c}{over $\F_2$}                                               & \multicolumn{3}{c}{over $\F_3$}                                                                                                                                                                                       \\
    \cmidrule(lr){2-4}\cmidrule(lr){5-7}
    $N$ & prev LB                                                                       & our LB                          & prev UB                                      & prev LB                                                                      & our LB       & prev UB                                \\
    \midrule
    $1$ & $1$                                                                           & $1$                             & $1$                                          & $1$                                                                          & $1$          & $1$                                    \\
    $2$ & $3$~\cite{fiduccia1977algebras}                                               & $3$                             & $3$~\cite{toom1963complexity}                & $3$~\cite{fiduccia1977algebras}                                              & $3$          & $3$~\cite{toom1963complexity}          \\
    $3$ & $6$~\cite{lempel1977new,brockett1978optimal,griesmer1960bound}     & $6$                             & $6$~\cite{kaminski1989multiplicative}        & $6$~\cite{lempel1977new,brockett1978optimal,grassl2007codetables} & $6$          & $6$~\cite{kaminski1989multiplicative}  \\
    $4$ & $9$~\cite{kaminski1989multiplicative}                                         & $9$                             & $9$~\cite{kaminski1989multiplicative}        & $9$~\cite{kaminski1989multiplicative}                                        & $9$          & $9$~\cite{kaminski1989multiplicative}  \\
    $5$ & $13$~\cite{lempel1977new,brockett1978optimal,grassl2007codetables} & $13$                            & $13$~\cite{montgomery2005five}               & $12$~\cite{kaminski1989multiplicative}                                       & $12$         & $12$~\cite{kaminski1989multiplicative} \\
    $6$ & $15$~\cite{lempel1977new,brockett1978optimal,grassl2007codetables} & $\boxed{16}$                    & $17$~\cite{montgomery2005five}               & $12$~\cite{lempel1977new,brockett1978optimal,solomon1965algebraically}   & $\boxed{14}$ & $15$~\cite{barbulescu2012finding}           \\
    $7$ & $18$~\cite{lempel1977new,brockett1978optimal,grassl2007codetables} & $\boxed{19}$                    & $22$~\cite{montgomery2005five}               &                                                                              &              &                                        \\
    $8$ & $20$~\cite{lempel1977new,brockett1978optimal,grassl2007codetables} & $\boxed{21}$                    & $26$~\cite{fan2007comments} &                                                                              &              &                                        \\
    \bottomrule
  \end{tabular}
  \caption{Full product of two degree-$(N-1)$ polynomials; boxed entries are new.}
  \label{tab:full}
\end{table}

\begin{table}[!ht]
  \centering
  \setlength{\tabcolsep}{6pt}
  \begin{tabular}{lllllll}
    \toprule
    & \multicolumn{3}{c}{over $\F_2$}                             & \multicolumn{3}{c}{over $\F_3$}                                                                                                                                                                                                            \\
    \cmidrule(lr){2-4}\cmidrule(lr){5-7}
    $N$  & prev LB                                                     & our LB                          & prev UB                                                      & prev LB                                                     & our LB       & prev UB                                                      \\
    \midrule
    $1$  & $1$                                                         & $1$                             & $1$                                                          & $1$                                                         & $1$          & $1$                                                          \\
    $2$  & $3$~\cite{winograd1977some}                                 & $3$                             & $3$~\cite{barbulescu2012finding}                         & $2$~\cite{winograd1977some}                                 & $2$          & $2$~\cite{wagh1983structured}                                \\
    $3$  & $4$~\cite{winograd1977some}                                 & $4$                             & $4$~\cite{wagh1983structured}                                & $5$~\cite{winograd1977some}                                 & $5$          & $5$~\cite{barbulescu2012finding}                         \\
    $4$  & $8$~\cite{blaser2005complete,barbulescu2012finding}         & $8$                             & $8$~\cite{barbulescu2012finding}  & $5$~\cite{winograd1977some}                                 & $5$          & $5$~\cite{wagh1983structured}                                \\
    $5$  & $10$~\cite{blaser2005complete,barbulescu2012finding}      & $9$                             & $10$~\cite{wagh1983structured}                               & $10$~\cite{blaser2005complete,barbulescu2012finding}        & $9$          & $10$~\cite{wagh1983structured}                               \\
    $6$  & $11$~\cite{averbuch1988classification,blaser2005complete}   & $11$                            & $12$~\cite{barbulescu2012finding}                        & $10$~\cite{winograd1977some}                                & $10$         & $10$~\cite{barbulescu2012finding}                        \\
    $7$  & $12$~\cite{degroote1983characterization,blaser2005complete} & $\boxed{13}$                    & $13$~\cite{wagh1983structured}                               & $13$~\cite{degroote1983characterization,blaser2005complete} & $13$         & $16$~\cite{cenk2010multiplication} \\
    $8$  & $16$~\cite{averbuch1991classification,blaser2005complete}   & $\boxed{19}$                    & $22$~\cite{cenk2011multiplication} & $11$~\cite{winograd1977some}                                & $11$         & $11$~\cite{wagh1983structured}                               \\
    $9$  & $19$~\cite{lempel1977new,grassl2007codetables,blaser2005complete}      & $18$                            & $19$~\cite{cenk2010multiplication}                        & $18$~\cite{averbuch1991classification,blaser2005complete}   & $\boxed{19}$ & $27$~\cite{cenk2011multiplication} \\
    $10$ & $19$~\cite{averbuch1988classification,blaser2005complete}   & $\boxed{22}$                    & $27$~\cite{randriambololona2012bilinear}                        &                                                             &              &                                                              \\
    \bottomrule
  \end{tabular}
  \caption{Cyclic convolution modulo $x^N-1$; boxed entries are new.}
  \label{tab:cyc}
\end{table}

\begin{table}[!ht]
  \centering
  \setlength{\tabcolsep}{6pt}
  \begin{tabular}{lllllll}
    \toprule
    & \multicolumn{3}{c}{over $\F_2$}                           & \multicolumn{3}{c}{over $\F_3$}                                                                                                                                                                                                  \\
    \cmidrule(lr){2-4}\cmidrule(lr){5-7}
    $N$  & prev LB                                                   & our LB                          & prev UB                                                  & prev LB                                                   & our LB       & prev UB                                                  \\
    \midrule
    $1$  & $1$                                                       & $1$                             & $1$                                                      & $1$                                                       & $1$          & $1$                                                      \\
    $2$  & $3$~\cite{winograd1977some,alder1981algorithmic}          & $3$                             & $3$~\cite{toom1963complexity}                            & $3$~\cite{winograd1977some,alder1981algorithmic}          & $3$          & $3$~\cite{toom1963complexity}                            \\
    $3$  & $5$~\cite{winograd1977some,alder1981algorithmic}          & $5$                             & $5$~\cite{averbuch1991classification,blaser2005complete} & $5$~\cite{winograd1977some,alder1981algorithmic}          & $5$          & $5$~\cite{averbuch1991classification,blaser2005complete} \\
    $4$  & $8$~\cite{averbuch1991classification,blaser2005complete}  & $8$                             & $8$~\cite{cenk2011multiplication}                        & $8$~\cite{averbuch1991classification,blaser2005complete}  & $7$          & $8$~\cite{cenk2011multiplication}                        \\
    $5$  & $11$~\cite{covanov2019improved}                           & $10$                            & $11$~\cite{cenk2011multiplication}                       & $10$~\cite{averbuch1991classification,blaser2005complete} & $10$         & $10$~\cite{barbulescu2012finding}                        \\
    $6$  & $12$~\cite{averbuch1991classification,blaser2005complete} & $\boxed{13}$                    & $14$~\cite{cenk2011multiplication}                       & $12$~\cite{averbuch1991classification,blaser2005complete} & $12$         & $14$~\cite{cenk2011multiplication}                       \\
    $7$  & $14$~\cite{averbuch1991classification,blaser2005complete} & $\boxed{16}$                    & $18$~\cite{cenk2011multiplication}                       & $14$~\cite{averbuch1991classification,blaser2005complete} & $14$         & $18$~\cite{cenk2011multiplication}                       \\
    $8$  & $16$~\cite{averbuch1991classification,blaser2005complete} & $\boxed{19}$                    & $22$~\cite{cenk2011multiplication}                       & $16$~\cite{averbuch1991classification,blaser2005complete} & $\boxed{17}$ & $22$~\cite{cenk2011multiplication}                       \\
    $9$  & $18$~\cite{averbuch1991classification,blaser2005complete} & $\boxed{21}$                    & $27$~\cite{cenk2011multiplication}$^\ast$                & $18$~\cite{averbuch1991classification,blaser2005complete} & $\boxed{19}$ & $27$~\cite{cenk2011multiplication}$^\ast$                      \\
    $10$ & $20$~\cite{averbuch1991classification,blaser2005complete} & $\boxed{24}$                    & $31$~\cite{cenk2011multiplication}                       & $20$~\cite{averbuch1991classification,blaser2005complete} & $\boxed{21}$ & $31$~\cite{cenk2011multiplication}                             \\
    $11$ & $22$~\cite{averbuch1991classification,blaser2005complete} & $\boxed{27}$                    & $36$~\cite{cenk2011multiplication}                       &                                                           &              &                                                          \\
    \bottomrule
  \end{tabular}

  \smallskip
  {\footnotesize\raggedright
    $^\ast$ We obtain a new bilinear complexity upper bound of $26$ for multiplication over $\mathbb{Z}[x]/x^9$ by flip-graph search over $\F_2$~\cite{kauers2023flip} and lifting to $\mathbb{Z}$.
    \scriptsize
    $(5a1-6a2+2a3+a5-a6)\otimes (3b0+2b1+b2)\otimes (-c6-c7-c8) + a0\otimes b0\otimes (c0-c3-c5-c8) + (a1-a2)\otimes (3b0+2b1+b2)\otimes (-c2-c3-c4-c6) + (a0-a1)\otimes (b0+b1+3b2+3b3+b4+b6+b7)\otimes (-c7-c8) + (2a0-a1)\otimes (b0-b5+b7)\otimes c7 + (3a0-2a1+a2-2a3+a4)\otimes (4b0+3b1+b2+b3+b4)\otimes c4 + (a2-4a3+2a4-2a5+3a6-a7)\otimes (b0+b1)\otimes c7 + (a1-a2)\otimes (3b0+b1-b2-3b3-3b4-2b5-b6)\otimes (c6+c7+c8) + (a2-a3+a4-2a5+a6)\otimes (2b0+2b1+b2)\otimes (-c6-c7) + (a0-a2+a3)\otimes (b0+b1+b4+b5)\otimes (-c3+c5+c8) + (2a0-a1)\otimes (b0+b1)\otimes (c1-c2-c3-c4-c5) + (a0-2a1+2a2-a3)\otimes (6b0+4b1+2b2+b3+b4+b5)\otimes (-c3-c6-c7-c8) + (a0+2a1-a2-2a3+2a4-a5)\otimes (b0+b1+b2)\otimes c6 + a0\otimes (b1+b2+b3+b4+b6+b7+b8)\otimes c8 + (a0-a1+a2-2a3+a4)\otimes (5b0+4b1+b2+b3+b4)\otimes (-c4-c7) + (a0-a1-a2-a3+a4)\otimes (5b0+3b1+2b2+2b3+b4)\otimes (-c4-c5) + (a0+a1-a2)\otimes (3b0+2b1+b2+b3+b4+b5+b6)\otimes (c6+c7) + (6a0-3a1-2a2-a3+2a6-a7)\otimes (2b0+b1)\otimes (-c7-c8) + (6a0-a1+a3-a4+a6-2a7+a8)\otimes b0\otimes c8 + (a2-a3)\otimes (b2-b4-b5)\otimes (c3-c5-c6-c8) + (a2+a3-a4)\otimes (7b0+4b1+2b2+2b3+b4)\otimes (-c4-c5-c7-c8) + (2a0-2a1+2a2-a3)\otimes (4b0+3b1+b2+b3+b4+b5)\otimes (c3+c7) + (a0+a1-a2)\otimes (2b0+2b1+b2)\otimes (c2+c7+c8) + (a0+a3-2a4+a5)\otimes (2b0+2b1+b2+b3)\otimes (-c5-c7) + (3a0-2a1+a2+a3-2a4+a5)\otimes (3b0+2b1+b2+b3)\otimes (c5+c7+c8) + (a0-a1)\otimes (2b0+b1)\otimes (-c1+c3+c6+c7+c8)$
  \par}
  \caption{Truncated product modulo $x^N$; boxed entries are new lower bounds.}
  \label{tab:trunc}
\end{table}

\begin{table}[!ht]
  \centering
  \setlength{\tabcolsep}{12pt}
  \begin{tabular}{llll}
    \toprule
    & \multicolumn{3}{c}{over $\F_3$}                                                                                                       \\
    \cmidrule(lr){2-4}
    $N$ & prev LB                                                     & our LB       & prev UB                                                  \\
    \midrule
    $1$ & $1$                                                         & $1$          & $1$                                                      \\
    $2$ & $3$~\cite{winograd1977some}                                 & $3$          & $3$~\cite{toom1963complexity}                            \\
    $3$ & $5$~\cite{winograd1977some}                                 & $5$          & $5$~\cite{averbuch1991classification,blaser2005complete} \\
    $4$ & $6$~\cite{winograd1977some}                                 & $6$          & $6$~\cite{degroote1983characterization}                  \\
    $5$ & $10$~\cite{blaser2005complete,barbulescu2012finding}        & $9$          & $10$~\cite{cenk2010multiplication}                       \\
    $6$ & $12$~\cite{averbuch1988classification,blaser2005complete}   & $12$         & $15$~\cite{barbulescu2012finding}                             \\
    $7$ & $13$~\cite{degroote1983characterization,blaser2005complete} & $13$         & $16$~\cite{cenk2010multiplication}                       \\
    $8$ & $16$~\cite{blaser2005complete,barbulescu2012finding}        & $16$         & $18$~\cite{cenk2010multiplication}                       \\
    $9$ & $18$~\cite{averbuch1991classification,blaser2005complete}   & $\boxed{19}$ & $27$~\cite{cenk2011multiplication}                       \\
    \bottomrule
  \end{tabular}
  \caption{Negacyclic convolution modulo $x^N+1$; the boxed entry is
    new. Over $\F_2$, $x^N+1=x^N-1$, so the negacyclic product coincides with cyclic
  convolution (Table~\ref{tab:cyc}).}
  \label{tab:nega}
\end{table}
\clearpage

\section{Computational results}\label{app:results}

Table~\ref{tab:results} reports the search and verification details for each new bound.
The bounds themselves and their comparison to prior work appear in
Tables~\ref{tab:matrix}--\ref{tab:nega} of Appendix~\ref{app:bound-tables}.

\begin{table}[!ht]
  \centering
  \setlength{\tabcolsep}{4pt}
  \begin{tabular}{lllcllc}
    \toprule
    result                                    & machine$^\dagger$        & \shortstack{search\\time} & \shortstack{search\\cost} & \shortstack{cert.\\size} & \shortstack{verify\\time} & \shortstack{verify\\cost} \\
    \midrule
    $\Rk{\F_2}{\mmt{3}{3}{3}}\ge 20$          & MacBook Air              & $40$ min                  & --                        & $32$ MiB                 & $3$ sec                   & --                        \\
    $\Rk{\F_2}{\mmt{3}{3}{4}}\ge 25$          & c8g.\{16,48\}xlarge$^\S$ & $4.6$ days                & $52$ USD                  & $600$ MiB                & $70$ min                  & $1$ USD                   \\
    $\Rk{\F_2}{\mmt{3}{4}{4}}\ge 29$          & c8g.48xlarge             & $4.9$ hrs                 & $5$ USD                   & $500$ MiB                & $3.4$ hrs                 & $3$ USD                   \\
    \midrule
    $\Rk{\F_2}{\Fullt{6}}\ge 16$              & c8g.48xlarge             & $11$ min                  & --                        & --                       & $10$ min                  & --                        \\
    $\Rk{\F_2}{\Fullt{7}}\ge 19$              & c8g.48xlarge             & $3$ hrs                   & $3$ USD                   & $9$ MiB                  & $2.6$ hrs                 & $3$ USD                   \\
    $\Rk{\F_2}{\Fullt{8}}\ge 21$              & c8g.48xlarge             & $6$ hrs                   & $6$ USD                   & $224$ MiB                & $3$ hrs                   & $3$ USD                   \\
    $\Rk{\F_3}{\Fullt{6}}\ge 14$              & c8g.48xlarge             & $35$ min                  & --                        & $6$ MiB                  & $31$ min                  & --                        \\
    \midrule
    $\Rk{\F_2}{\Cyct{7}}\ge 13$               & c8g.48xlarge             & --                        & --                        & --                       & --                        & --                        \\
    $\Rk{\F_2}{\Cyct{8}}\ge 19$               & c8g.48xlarge             & --                        & --                        & --                       & --                        & --                        \\
    $\Rk{\F_2}{\Cyct{10}}\ge 22$              & c8g.48xlarge             & --                        & --                        & $8$ MiB                  & --                        & --                        \\
    $\Rk{\F_3}{\Cyct{9}}\ge 19$               & c8g.48xlarge             & $3$ min                   & --                        & $3$ MiB                  & --                        & --                        \\
    \midrule
    $\Rk{\F_2}{\Trunct{6}}\ge 13$             & c8g.48xlarge             & --                        & --                        & --                       & --                        & --                        \\
    $\Rk{\F_2}{\Trunct{7}}\ge 16$             & c8g.48xlarge             & --                        & --                        & --                       & --                        & --                        \\
    $\Rk{\F_2}{\Trunct{8}}\ge 19$             & c8g.48xlarge             & --                        & --                        & --                       & --                        & --                        \\
    $\Rk{\F_2}{\Trunct{9}}\ge 21$             & c8g.48xlarge             & --                        & --                        & --                       & --                        & --                        \\
    $\Rk{\F_2}{\Trunct{10}}\ge 24$            & c8g.48xlarge             & $1$ min                   & --                        & $3$ MiB                  & --                        & --                        \\
    $\Rk{\F_2}{\Trunct{11}}\ge 27$            & c8g.48xlarge             & $7$ min                   & --                        & $22$ MiB                 & --                        & --                        \\
    $\Rk{\F_3}{\Trunct{8}}\ge 17$             & c8g.48xlarge             & --                        & --                        & --                       & --                        & --                        \\
    $\Rk{\F_3}{\Trunct{9}}\ge 19$             & c8g.48xlarge             & $2$ min                   & --                        & $2$ MiB                  & --                        & --                        \\
    $\Rk{\F_3}{\Trunct{10}}\ge 21$            & c8g.48xlarge             & $3.4$ hrs                 & $3$ USD                   & $20$ MiB                 & --                        & --                        \\
    \midrule
    $\Rk{\F_3}{\Negat{9}}\ge 19$              & c8g.48xlarge             & $3$ min                   & --                        & $2$ MiB                  & --                        & --                        \\
    \bottomrule
  \end{tabular}

  \smallskip
  {\footnotesize\raggedright
    -- Negligible.\\
    $^\dagger$ Machines: Apple MacBook Air M4 with 16 GB RAM; AWS c8g.16xlarge Spot
    Instance, Graviton4 processors, 64 vCPUs, 128 GiB memory; AWS c8g.48xlarge Spot
    Instance, Graviton4 processors, 192 vCPUs, 384 GiB memory.\\
    $^\S$ Search on c8g.16xlarge. Verification on c8g.48xlarge.\\
    \par
  }
  \caption{Search and verification details for each reported lower bound.}
  \label{tab:results}
\end{table}

\clearpage
\section{Toy example in detail: \texorpdfstring{$\Rk{\F_2}{\mmt{2}{2}{2}}\ge 7$}{rank over F2 of <2,2,2> >= 7}}\label{sec:toy-in-detail}

We illustrate the dynamic program on $\mmt{2}{2}{2}$ over $\F_2$ in detail, where, as in
Section~\ref{sec:toy}, $A=\F_2^{2\times2}$ is the space of first matrix arguments,
$X\in A$, and the symmetry group is
$G=(\GL_2(\F_2)\times\GL_2(\F_2))\rtimes C_2$ acting by $X\mapsto PXQ^{-1}$ and by
transposition. There are ten orbits of subspaces. The program processes
them from the most constrained to the unconstrained tensor, assigning each a lower
bound by whichever of the four techniques of Section~\ref{sec:techniques} succeeds
first. We write $a_{ij}\in A^*$ for the coordinate functionals that survive the constraints
and $c_{ki}=e_{ik}\in C$ for the output matrix units; the constrained tensor is
$\sum_{i,j,k} a_{ij}\otimes b_{jk}\otimes c_{ki}$ restricted to them.

\subsection*{Orbit 0: $\{a_{00}=0,\ a_{01}=0,\ a_{10}=0,\ a_{11}=0\}$ (flatten)}
$X=
\begin{pmatrix}0 & 0 \\0&0
\end{pmatrix}$. The constrained tensor is $T_0=0$, so
$\R(T_0)=0$.

\subsection*{Orbit 1: $\{a_{00}=0,\ a_{01}=0,\ a_{10}=0\}$ (flatten)}
$X=
\begin{pmatrix}0 & 0 \\0&a_{11}
\end{pmatrix}$, with tensor
$T_1=a_{11}\otimes b_{10}\otimes c_{01}+a_{11}\otimes b_{11}\otimes c_{11}$.
Flattening the first two factors into one yields a matrix of rank $2$, so
$\R(T_1)\ge 2$.

\subsection*{Orbit 2: $\{a_{00}=0,\ a_{01}+a_{10}=0,\ a_{11}=0\}$ (flatten)}
$X=
\begin{pmatrix}0 & a_{01} \\a_{01}&0
\end{pmatrix}$, with tensor
$T_2=a_{01}\otimes b_{00}\otimes c_{01}+a_{01}\otimes b_{01}\otimes c_{11}
+a_{01}\otimes b_{10}\otimes c_{00}+a_{01}\otimes b_{11}\otimes c_{10}$.
Flattening yields a matrix of rank $4$, so $\R(T_2)\ge 4$.

\subsection*{Orbit 3: $\{a_{00}=0,\ a_{01}=0\}$ (flatten)}
$X=
\begin{pmatrix}0 & 0 \\a_{10}&a_{11}
\end{pmatrix}$, with tensor
$T_3=a_{10}\otimes b_{00}\otimes c_{01}+a_{10}\otimes b_{01}\otimes c_{11}
+a_{11}\otimes b_{10}\otimes c_{01}+a_{11}\otimes b_{11}\otimes c_{11}$.
Flattening yields rank $4$, so $\R(T_3)\ge 4$.

\subsection*{Orbit 4: $\{a_{00}=0,\ a_{01}+a_{10}=0\}$ (forced product)}
$X=
\begin{pmatrix}0 & a_{01} \\a_{01}&a_{11}
\end{pmatrix}$, with tensor
\[
  T_4=a_{01}\otimes b_{00}\otimes c_{01}+a_{01}\otimes b_{01}\otimes c_{11}
  +a_{01}\otimes b_{10}\otimes c_{00}+a_{01}\otimes b_{11}\otimes c_{10}
  +a_{11}\otimes b_{10}\otimes c_{01}+a_{11}\otimes b_{11}\otimes c_{11}.
\]
Flattening yields only $4$, so we invoke the forced-product technique
(Section~\ref{sec:forced-product}), which automates the substitution of Hopcroft and
Kerr stated in Lemma~\ref{lem:forced-product}.

Fix an optimal decomposition $T_4=\sum_{\lambda} u_\lambda\otimes v_\lambda\otimes
w_\lambda$ with $r=\R(T_4)$ terms. Grouping $T_4$ along the $C$ factor exposes its
slices as forms in $A^*\otimes B^*$: the slice at $c_{00}$ is the single product
$a_{01}\otimes b_{10}$ and the slice at $c_{10}$ is the single product
$a_{01}\otimes b_{11}$, and these two are linearly independent (so $s=2$). By
Lemma~\ref{lem:forced-product} we may assume the decomposition computes them
literally, say $u_0\otimes v_0=a_{01}\otimes b_{10}$ and
$u_1\otimes v_1=a_{01}\otimes b_{11}$.

Now strip these two terms from $T_4$. The lemma pins down their $A^*\otimes B^*$ parts
but not their $C$-components $w_0,w_1$: each must match the slice it computes along
$c_{00}$ (resp.\ $c_{10}$), but its coordinates along the remaining basis vectors
$c_{01},c_{11}$ are unconstrained. Over $\F_2$ this leaves four unknown bits
$\mu_0,\dots,\mu_3$, and the residual is
\[
  T_4'=T_4+a_{01}\otimes b_{10}\otimes(c_{00}+\mu_0 c_{01}+\mu_1 c_{11})
  +a_{01}\otimes b_{11}\otimes(c_{10}+\mu_2 c_{01}+\mu_3 c_{11}),
  \qquad \mu_i\in\{0,1\}.
\]
Enumerating all $2^4=16$ assignments of $\mu$ and flattening each gives $\R(T_4')\ge 4$ in
every case, so $\R(T_4)\ge 4+2=6$.

\subsection*{Orbit 5: $\{a_{00}=0,\ a_{11}=0\}$ (flatten)}
$X=
\begin{pmatrix}0 & a_{01} \\a_{10}&0
\end{pmatrix}$, with tensor
$T_5=a_{01}\otimes b_{10}\otimes c_{00}+a_{01}\otimes b_{11}\otimes c_{10}
+a_{10}\otimes b_{00}\otimes c_{01}+a_{10}\otimes b_{01}\otimes c_{11}$.
Flattening yields rank $4$, so $\R(T_5)\ge 4$.

\subsection*{Orbit 6: $\{a_{01}+a_{10}=0,\ a_{00}+a_{01}+a_{11}=0\}$ (substitution with backtracking)}
$X=
\begin{pmatrix}a_{00} & a_{01} \\a_{01}&a_{00}+a_{01}
\end{pmatrix}$. Flattening gives
only $4$. Here we demonstrate substitution with backtracking
(Section~\ref{sec:backtracking}).
Fix an optimal decomposition $T_6=\sum_\lambda u_\lambda\otimes v_\lambda\otimes
w_\lambda$; by Lemma~\ref{lem:substitution} we may take each $u_\lambda$ to be a nonzero element of
$\operatorname{span}\{a_{00},a_{01}\}$. The canonical forms are the nonzero linear forms in
$a_{00},a_{01}$ over $\F_2$, ordered $f_0=a_{00}$, $f_1=a_{01}$,
$f_2=a_{00}+a_{01}$; every nonzero $u_\lambda$ is equivalent to one of them.
Imposing any single form as an extra constraint lands in Orbit~2, of certified
bound $4$:
\begin{itemize}
  \item $a_{00}=0$ leaves $
    \begin{pmatrix}0 & a_{01} \\a_{01}&a_{01}
    \end{pmatrix}$,
    Orbit 2 after adding row $0$ to row $1$;
  \item $a_{01}=0$ leaves $
    \begin{pmatrix}a_{00} & 0 \\0&a_{00}
    \end{pmatrix}$,
    Orbit 2 after swapping the rows;
  \item $a_{00}+a_{01}=0$ leaves $
    \begin{pmatrix}a_{00} & a_{00} \\a_{00}&0
    \end{pmatrix}$,
    Orbit 2 after adding row $1$ to row $0$;
\end{itemize}
imposing two distinct forms forces $X=0$, Orbit~0, of bound $0$.

The search tree of the final run, with target $6$, is the following, where
``$6$ not met'' means that no subset of the
chain proves the target and ``Proved; backtrack'' marks a successful leaf.
\begin{itemize}
  \item $a_{00}$: $6$ not met; go deeper.
    \begin{itemize}
      \item $a_{00},a_{00}$: set $a_{00}=0$, reaching Orbit 2; $\R(T_6)\ge\R(T_2)+2\ge 6$. Proved; backtrack.
      \item $a_{00},a_{01}$: $6$ not met; go deeper.
        \begin{itemize}
          \item $a_{00},a_{01},a_{01}$: set $a_{01}=0$, reaching Orbit 2; $\R(T_6)\ge 6$. Proved; backtrack.
          \item $a_{00},a_{01},a_{00}+a_{01}$: $6$ not met; go deeper.
            \begin{itemize}
              \item $a_{00},a_{01},a_{00}+a_{01},a_{00}+a_{01}$: set $a_{00}+a_{01}=0$, reaching Orbit 2; $\R(T_6)\ge 6$. Proved; backtrack.
            \end{itemize}
        \end{itemize}
      \item $a_{00},a_{00}+a_{01}$: $6$ not met; go deeper.
        \begin{itemize}
          \item $a_{00},a_{00}+a_{01},a_{00}+a_{01}$: set $a_{00}+a_{01}=0$, reaching Orbit 2; $\R(T_6)\ge 6$. Proved; backtrack.
        \end{itemize}
    \end{itemize}
  \item $a_{01}$: $6$ not met; go deeper.
    \begin{itemize}
      \item $a_{01},a_{01}$: set $a_{01}=0$, reaching Orbit 2; $\R(T_6)\ge 6$. Proved; backtrack.
      \item $a_{01},a_{00}+a_{01}$: $6$ not met; go deeper.
        \begin{itemize}
          \item $a_{01},a_{00}+a_{01},a_{00}+a_{01}$: set $a_{00}+a_{01}=0$, reaching Orbit 2; $\R(T_6)\ge 6$. Proved; backtrack.
        \end{itemize}
    \end{itemize}
  \item $a_{00}+a_{01}$: $6$ not met; go deeper.
    \begin{itemize}
      \item $a_{00}+a_{01},a_{00}+a_{01}$: set $a_{00}+a_{01}=0$, reaching Orbit 2; $\R(T_6)\ge 6$. Proved; backtrack.
    \end{itemize}
\end{itemize}
Every branch reaches the target, so $\R(T_6)\ge 6$.

\subsection*{Orbit 7: $\{a_{00}=0\}$ (degenerate reduction)}
$X=
\begin{pmatrix}0 & a_{01} \\a_{10}&a_{11}
\end{pmatrix}$. Adding the constraint
$a_{01}+a_{10}=0$ lands in Orbit 4, so degenerate reduction
(Section~\ref{sec:degenerate}) gives $\R(T_7)\ge\R(T_4)\ge 6$.

\subsection*{Orbit 8: $\{a_{01}+a_{10}=0\}$ (degenerate reduction)}
$X=
\begin{pmatrix}a_{00} & a_{01} \\a_{01}&a_{11}
\end{pmatrix}$. Adding the constraint
$a_{00}=0$ lands in Orbit 4, so $\R(T_8)\ge\R(T_4)\ge 6$.

\subsection*{Orbit 9: $\{\,\}$ (substitution with backtracking)}
$X=
\begin{pmatrix}a_{00} & a_{01} \\a_{10}&a_{11}
\end{pmatrix}$ is the unconstrained
tensor $T_9=\sum_{i,j,k}a_{ij}\otimes b_{jk}\otimes c_{ki}=\mmt{2}{2}{2}$. We apply
substitution with backtracking. Fix an optimal decomposition
$T_9=\sum_\lambda u_\lambda\otimes v_\lambda\otimes w_\lambda$. If $a_{00}$ appears
as some $u_\lambda$, setting $a_{00}=0$ lands in Orbit 7 and gives
$\R(T_9)\ge\R(T_7)+1\ge 7$; by the sandwich symmetry the same holds when the form
is any of $a_{01}$, $a_{10}$, $a_{11}$, $a_{00}+a_{01}$, $a_{10}+a_{11}$,
$a_{00}+a_{10}$, $a_{01}+a_{11}$, or $a_{00}+a_{01}+a_{10}+a_{11}$. Otherwise the
form is one of $a_{01}+a_{10}$, $a_{00}+a_{11}$, $a_{00}+a_{01}+a_{10}$,
$a_{00}+a_{01}+a_{11}$, $a_{00}+a_{10}+a_{11}$, or $a_{01}+a_{10}+a_{11}$, and
setting it to zero lands in Orbit 8, giving $\R(T_9)\ge\R(T_8)+1\ge 7$. These
exhaust the nonzero forms, so $\Rk{\F_2}{\mmt{2}{2}{2}}=\R(T_9)\ge 7$.

The search tree here is only one level deep: whichever form is substituted first
lands in Orbit 7 or Orbit 8, both certified at $6$, so a single substitution already
meets the target of $7$ and no branch has to go deeper.

\section{Orbit-enumeration algorithm}\label{app:enum}

Algorithm~\ref{alg:enum} spells out the orbit enumerator of
Section~\ref{sec:orbits}: each codimension-$d^{\perp}$ representative is built by extending
a codimension-$(d^{\perp}-1)$ one with a new highest-pivot constraint (representatives are
manipulated as the RREFs of their constraint sets), and the inner
duplicate test is the meet-in-the-middle probe of that section. The
$\textit{visited}$ set is cleared at the start of each codimension, since the
codimension is an invariant of the group action.

\begin{algorithm}[!ht]
  \caption{Enumerate subspace orbits}
  \label{alg:enum}
  \begin{algorithmic}[1]
    \State $\textit{reps}[0]\gets[\varnothing]$
    \For{$d^{\perp}=1$ \textbf{to} $N_A$}
    \State $\textit{visited}\gets\varnothing$
    \For{each representative $c\in \textit{reps}[d^{\perp}-1]$ (in lex order)}
    \For{each candidate constraint $r$ above all pivots of $c$ (in lex order)}
    \State $q \gets c \cup \{r\}$
    \If{\textbf{not} \Call{Seen}{$q$}}
    \State append $q$ to $\textit{reps}[d^{\perp}]$
    \For{$\sigma\in\Sigma$}\ \ insert $\RREF(\sigma(q))$ into \textit{visited}\EndFor
    \EndIf
    \EndFor
    \EndFor
    \EndFor
    \State \Return \textit{reps}
    \Statex
    \Function{Seen}{$q$}
    \For{$\kappa\in\Kappa$}
    \If{$\RREF(\kappa(q))\in\textit{visited}$}\ \Return \textbf{true}\EndIf
    \EndFor
    \State \Return \textbf{false}
    \EndFunction
  \end{algorithmic}
\end{algorithm}

\subsection*{Orbit counts}\label{app:counts}

Throughout this subsection we take matrix multiplication as the example family;
the polynomial families are swept by the same procedure.
Table~\ref{tab:orbits} reports the number of subspace orbits Algorithm~\ref{alg:enum} finds
for each matrix multiplication format over $\F_2$ and $\F_3$. The count depends on whether the transpose
symmetry is in force, which happens exactly when the format is square ($l=m=n$).

\begin{table}[!ht]
  \centering
  \setlength{\tabcolsep}{12pt}
  \begin{tabular}{llrr}
    \toprule
    format $l\times m$ & transpose? & \# orbits over $\F_2$ & \# orbits over $\F_3$          \\
    \midrule
    $2\times 2$        & yes        & $10$                  & $10$                           \\
    $2\times 2$        & no         & $11$                  & $11$                           \\
    $2\times 3$        & no         & $31$                  & $31$                           \\
    $2\times 4$        & no         & $86$                  & $91$                           \\
    $3\times 3$        & yes        & $496$                 & $736$                          \\
    $3\times 3$        & no         & $710$                 & $1046$                         \\
    $3\times 4$        & no         & $158{,}426$           & ${>}1.6\times10^{6}\,^\dagger$ \\
    $4\times 4$        & yes        & ${>}1.6\times10^{11}$ & ${>}8.8\times10^{15}$          \\
    \bottomrule
  \end{tabular}

  \smallskip
  {\footnotesize\raggedright
    $^\dagger$ The exact value is computable but unlikely to improve any matrix
  bound.\par}
  \caption{Number of subspace orbits for matrix multiplication,
    $A=\F_q^{l\times m}$. ``transpose?'' marks whether the order-two transpose
    symmetry is included; it applies only to square formats $l=m=n$. The ${>}$
  entries follow from the exact lower bound in Equation~\eqref{eq:orbit-count}.}
  \label{tab:orbits}
\end{table}

Where the exact orbit count is out of reach, an exact counting expression still
gives a rigorous lower bound. There are
$\genfrac{[}{]}{0pt}{}{N}{d}_q$ dimension-$d$ subspaces of $\F_q^N$, where
$\genfrac{[}{]}{0pt}{}{N}{d}_q$ is the Gaussian binomial coefficient. Since
each orbit has size at most the group order,
\begin{equation}\label{eq:orbit-count}
  \#\bigl(\operatorname{Sub}(\F_q^{lm})/G\bigr)
  \;\ge\;
  \frac{\displaystyle\sum_{d=0}^{lm}\genfrac{[}{]}{0pt}{}{lm}{d}_q}
  {|\GL_l(\F_q)|\,|\GL_m(\F_q)|\,|H|}
  ,
  \qquad
  \genfrac{[}{]}{0pt}{}{N}{d}_q
  =\prod_{i=0}^{d-1}\frac{q^{N-i}-1}{q^{d-i}-1},
\end{equation}
where $H=C_2$ when the transpose symmetry is present ($l=m=n$) and $H=C_1$
otherwise, and $|\GL_n(\F_q)|=\prod_{k=0}^{n-1}(q^n-q^k)$.
This produces the
${>}$ entries of Table~\ref{tab:orbits}. Both the estimate and the
enumeration can be nearly halved by duality: orthogonal
complement matches each orbit of dimension-$d$ subspaces with one of dimension
$lm-d$, so only the
dimensions $d\le\lfloor lm/2\rfloor$ need be enumerated.

The $4\times 4$ row makes the scaling barrier concrete: even over $\F_2$ there are
more than $10^{11}$ orbits, so merely enumerating them is already out of reach, let
alone running the dynamic program on each. This orbit explosion
is the main obstacle to applying the method to $\mmt{4}{4}{4}$.

\subsection*{More invariants}\label{app:invariants}

Clearing $\textit{visited}$ per codimension (equivalently, per dimension) exploits a single $G$-invariant, the
dimension: two subspaces in the same orbit always have the same
dimension, so the set never needs to hold subspaces of different dimensions
simultaneously. The same idea applies to \emph{any} $G$-invariant. Fixing a value
(or a tuple of values) of one or more invariants, keeping in $\textit{visited}$
only the representatives realizing that value, and sweeping the values one at a
time replaces a single large hash set with several small ones and never merges
orbits across values, since equivalent subspaces agree on every invariant. We do
not currently exploit invariants beyond the dimension, as enumeration memory is
not the bottleneck of our results. Taking matrix multiplication as an example,
the invariants of $G=(\GL_l(\F_q)\times\GL_m(\F_q))\rtimes C_2$ acting on
$A=\F_q^{l\times m}$ include the following.

\begin{itemize}
  \item \textbf{Rank distribution:} $\rho_i(S)=|\{M\in S:\rank M=i\}|$ for
    $i=0,\dots,\min(l,m)$. The sandwich action $M\mapsto PMQ^{-1}$ and the
    transpose preserve the rank of each $M$, so $\rho$ is constant on orbits.
  \item \textbf{Point profiles:} let $B_0,\dots,B_{d-1}$ be a basis of $S$.
    The left profile is
    $\lambda_i(S)=|\{x\in\F_q^{l}:\rank[x^\top B_0,\dots,x^\top B_{d-1}]=i\}|$ for
    $i=0,\dots,m$, and the right profile is
    $\mu_i(S)=|\{y\in\F_q^{m}:\rank[B_0y,\dots,B_{d-1}y]=i\}|$ for $i=0,\dots,l$.
    Left multiplication preserves $\lambda$, right multiplication preserves $\mu$,
    and a change of basis of $S$ preserves both; when $l=m=n$ the transpose swaps
    them, so the unordered pair $\{\lambda,\mu\}$ is invariant.
\end{itemize}

\end{document}